\documentclass[11pt]{article}
\usepackage{amsmath, amssymb, amsfonts, amsthm}
\usepackage{epsfig}

\newtheorem{theorem}{Theorem}[section]
\newtheorem{proposition}[theorem]{Proposition}
\newtheorem{corollary}[theorem]{Corollary}

\newcommand{\rd}{{\rm d}}
\newcommand{\be}{\begin{equation}}
\newcommand{\ee}{\end{equation}}
\newcommand{\bey}{\begin{eqnarray}}
\newcommand{\eey}{\end{eqnarray}}

\newcommand{\eps}{\varepsilon}

\newcommand{\bx}{{\bf x}}

\newcommand{\ph}{\varphi}

\renewcommand{\a}{\alpha}

\newcommand{\cU}{{\cal U}}

\newcommand{\bR}{{\mathbb R}}

\newcommand{\bN}{{\mathbb N}}

\newcommand{\tr}{\mbox{Tr}}

\newcommand{\wh}{\widehat}

\newcommand{\const}{\mathrm{const}}

\newcommand{\cF}{{\cal F}}

\newcommand{\cE}{{\cal E}}

\newcommand{\cK}{{\cal K}}
\newcommand{\cH}{{\cal H}}
\newcommand{\cL}{{\cal L}}

\input epsf

\newcommand{\donothing}[1]{}

\begin{document}
\title{Dynamics of Bose-Einstein Condensates}
\author{Benjamin Schlein\\
\\
Department of Mathematics, University of California at Davis, CA
95616, USA}

\maketitle

\begin{abstract}
We report on some recent results concerning the dynamics of
Bose-Einstein condensates, obtained in a series of joint papers
\cite{ESY2,ESY3} with L. Erd\H os and H.-T. Yau. Starting from many
body quantum dynamics, we present a rigorous derivation of a cubic
nonlinear Schr\"odinger equation known as the Gross-Pitaevskii
equation for the time evolution of the condensate wave function.
\end{abstract}

\section{Introduction}
\setcounter{equation}{0}

Bosonic systems at very low temperature are characterized by the
fact that a  macroscopic fraction of the particles collapses into a
single one-particle state. Although this phenomenon, known as
Bose-Einstein condensation, was already predicted in the early days
of quantum mechanics, the first empirical evidence for its existence
was only obtained in 1995, in experiments performed by groups led by
Cornell and Wieman at the University of Colorado at Boulder and by
Ketterle at MIT (see \cite{CW,Kett}). In these important
experiments, atomic gases were initially trapped by magnetic fields
and cooled down at very low temperatures. Then the magnetic traps
were switched off and the consequent time evolution of the gas was
observed; for sufficiently small temperatures, the particles
remained close together and the gas moved as a single particle, a
clear sign for the existence of condensation.

\medskip

In the last years important progress has also been achieved in the
theoretical understanding of Bose-Einstein condensation. In
\cite{LSY}, Lieb, Yngvason, and Seiringer considered a trapped Bose
gas consisting of $N$ three-dimensional particles described by the
Hamiltonian
\begin{equation}\label{eq:ham1} H^{\text{trap}}_N
= \sum_{j=1}^N \left(-\Delta_j +
V_{\text{ext}} (x_j) \right) + \sum_{i<j}^N V_a (x_i -x_j),
\end{equation} where $V_{\text{ext}}$ is an external confining
potential and $V_a (x)$ is a repulsive interaction potential with
scattering length $a$ (here and in the rest of the paper we use the
notation $\nabla_j = \nabla_{x_j}$ and $\Delta_j = \Delta_{x_j}$).
Letting $N \to \infty$ and $a \to 0$ with $Na=a_0$ fixed, they
showed that the ground state energy $E (N)$ of (\ref{eq:ham1})
divided by the number of particle $N$ converges to
\[ \lim_{N \to \infty, \; Na = a_0}\frac{E(N)}{N} =
\min_{\ph \in L^2 (\bR^3): \; \| \ph\| =1} \cE_{\text{GP}} (\ph) \]
where $\cE_{\text{GP}}$ is the Gross-Pitaevskii energy functional
\begin{equation}\label{eq:GPfunc} \cE_{\text{GP}} (\ph) = \int \rd x \; \left( |\nabla \ph
(x)|^2 + V_{\text{ext}} (x) |\ph (x)|^2 + 4 \pi a_0 |\ph (x)|^4
\right) \, .\end{equation} Later, in \cite{LS}, Lieb and Seiringer
also proved that trapped Bose gases characterized by the
Gross-Pitaevskii scaling $Na = a_0 = \const$ exhibit Bose-Einstein
condensation in the ground state. More precisely, they showed that,
if $\psi_N$ is the ground state wave function of the Hamiltonian
(\ref{eq:ham1}) and if $\gamma^{(1)}_N$ denotes the corresponding
one-particle marginal (defined as the partial trace of the density
matrix $\gamma_N = |\psi_N \rangle \langle \psi_N |$ over the last
$N-1$ particles, with the convention that $\tr \; \gamma^{(1)}_N =
1$ for all $N$), then
\begin{equation}\label{eq:conden}
\gamma_N^{(1)} \to |\phi_{\text{GP}} \rangle \langle
\phi_{\text{GP}}| \qquad \text{as } N \to \infty \, .
\end{equation}
Here $\phi_{\text{GP}} \in L^2 (\bR^3)$ is the minimizer of the
Gross-Pitaevskii energy functional (\ref{eq:GPfunc}). The
interpretation of this result is straightforward; in the limit of
large $N$, all particles, apart from a fraction vanishing as $N \to
\infty$, are in the same one-particle state described by the
wave-function $\phi_{\text{GP}} \in L^2 (\bR^3)$. In this sense the
ground state of (\ref{eq:ham1}) exhibits complete Bose-Einstein
condensation into $\phi_{\text{GP}}$.

\medskip

In joint works with L. Erd\H os and H.-T. Yau (see
\cite{ESY2,ESY3,ESY4}), we prove that the Gross-Pitaevskii theory
can also be used to describe the dynamics of Bose-Einstein
condensates. In the Gross-Pitaevskii scaling (characterized by the
fact that the scattering length of the interaction potential is of
the order $1/N$) we show, under some conditions on the interaction
potential and on the initial $N$-particle wave function, that
complete Bose-Einstein condensation is preserved by the time
evolution. Moreover we prove that the dynamics of the condensate
wave function is governed by the time-dependent Gross-Pitaevskii
equation associated with the energy functional (\ref{eq:GPfunc}).

\medskip

As an example, consider the experimental set-up described above,
where the dynamics of an initially confined gas is observed after
removing the traps. Mathematically, the trapped gas can be described
by the Hamiltonian (\ref{eq:ham1}), where the confining potential
$V_{\text{ext}}$ models the magnetic traps. When cooled down at very
low temperatures, the system essentially relaxes to the ground state
$\psi_N$ of (\ref{eq:ham1}); from \cite{LS} it follows that at time
$t=0$, immediately before switching off the traps, the system
exhibits complete Bose-Einstein condensation into $\phi_{\text{GP}}$
in the sense (\ref{eq:conden}). At time $t=0$ the traps are turned
off, and one observes the evolution of the system generated by the
translation invariant Hamiltonian
\[ H_N = -\sum_{j=1}^N \Delta_j + \sum_{i<j}^N V_a (x_i -x_j) \, .\]
Our results (stated in more details in Section \ref{sec:main} below)
imply that, if $\psi_{N,t} = e^{-i H_N t} \psi_N$ is the time
evolution of the initial wave function $\psi_N$ and if
$\gamma_{N,t}^{(1)}$ denotes the one-particle marginal associated
with $\psi_{N,t}$, then, for any fixed time $t \in \bR$,
\[ \gamma^{(1)}_{N,t} \to |\ph_t \rangle \langle \ph_t| \qquad
\text{as } N \to \infty \] where $\ph_t$ is the solution of the
nonlinear time-dependent Gross-Pitaevskii equation
\begin{equation}\label{eq:GP1}
i \partial_t \ph_t = - \Delta \ph_t + 8 \pi a_0 |\ph_t|^2 \ph_t \,
\end{equation}
with the initial data $\ph_{t=0} = \phi_{\text{GP}}$. In other
words, we prove that at arbitrary time $t \in \bR$, the system still
exhibits complete condensation, and the time-evolution of the
condensate wave function is determined by the Gross-Pitaevskii
equation (\ref{eq:GP1}).

\medskip

The goal of this manuscript is to illustrate the main ideas of the
proof of the results obtained in \cite{ESY2,ESY3,ESY4}. The paper is
organized as follows. In Section~\ref{sec:heu} we define the model
more precisely, and we give a heuristic argument to explain the
emergence of the Gross-Pitaevskii equation (\ref{eq:GP1}). In
Section~\ref{sec:main} we present our main results. In
Section~\ref{sec:proof} we illustrate the general strategy used to
prove the main results and, finally, in Sections~\ref{sec:conv} and
\ref{sec:unique} we discuss the two most important parts of the
proof in some more details.

\section{Heuristic Derivation of the Gross-Pitaevskii
Equation}\label{sec:heu} \setcounter{equation}{0}

To describe the interaction among the particles we choose a
positive, spherical symmetric, compactly supported, smooth function
$V(x)$. We denote the scattering length of $V$ by $a_0$.

\medskip

Recall that the scattering length of $V$ is defined by the spherical
symmetric solution to the zero energy equation
\begin{equation}\label{eq:0en} \left( -\Delta + \frac{1}{2} V(x) \right) f(x) = 0 \qquad
f(x) \to 1 \quad \text{as } |x| \to \infty\,. \end{equation} The
scattering length of $V$ is defined then by \[ a_0 = \lim_{|x| \to
\infty} |x| - |x| f(x)\,. \] This limit can be proven to exist if
$V$ decays sufficiently fast at infinity. Note that, since we
assumed $V$ to have compact support, we have
\begin{equation}\label{eq:f} f(x) = 1 - \frac{a_0}{|x|} \end{equation}
for $|x|$ sufficiently large. Another equivalent characterization of
the scattering length is given by \begin{equation}\label{eq:8pia0} 8
\pi a_0 = \int \rd x \, V(x) f(x) \, .\end{equation}

To recover the Gross-Pitaevskii scaling, we define \( V_N (x) = N^2
V (Nx)\). By scaling it is clear that the scattering length of $V_N$
equals $a = a_0 /N$. In fact if $f(x)$ is the solution to
(\ref{eq:0en}), it is clear that $f_N (x) = f(Nx)$ solves
\begin{equation}\label{eq:0en-N}
\left(-\Delta + \frac{1}{2} V_N (x) \right) f_N (x) = 0
\end{equation}
with the boundary condition $f_N (x) \to 1$ as $|x| \to \infty$.
{F}rom (\ref{eq:f}), we obtain \[ f_N (x) = 1 - \frac{a_0}{N|x|} = 1
- \frac{a}{|x|}  \] for $|x|$ large enough. In particular the
scattering length $a$ of $V_N$ is given by $a=a_0/N$.

\bigskip

We consider the dynamics generated by the translation invariant
Hamiltonian
\begin{equation}\label{eq:ham}
H_N = \sum_{j=1}^N -\Delta_j + \sum_{i<j}^N V_N (x_i -x_j)\,
\end{equation}
acting on the Hilbert space $L_s^2 (\bR^{3N}, \rd x_1 \dots \rd
x_N)$, the bosonic subspace of $L^2 (\bR^{3N}, \rd x_1 \dots \rd
x_N)$ consisting of all permutation symmetric functions (although it
is possible to extend our analysis to include an external potential,
to keep the discussion as simple as possible we only consider the
translation invariant case (\ref{eq:ham})). We consider solutions
$\psi_{N,t}$ of the $N$-body Schr\"odinger equation
\begin{equation}\label{eq:schr}
i\partial_t \psi_{N,t} = H_N \psi_{N,t}\,.
\end{equation}
Let $\gamma_{N,t} = |\psi_{N,t} \rangle \langle \psi_{N,t}|$ denote
the density matrix associated with $\psi_{N,t}$, defined as the
orthogonal projection onto $\psi_{N,t}$. In order to study the limit
$N \to \infty$, we introduce the marginal densities of
$\gamma_{N,t}$. For $k=1, \dots,N$, we define the $k$-particle
density matrix $\gamma_{N,t}^{(k)}$ associated with $\psi_{N,t}$ by
taking the partial trace of $\gamma_{N,t}$ over the last $N-k$
particles. In other words, $\gamma_{N,t}^{(k)}$ is defined as the
positive trace class operator on $L_s^2 (\bR^{3k})$ with kernel
given by
\begin{equation}\label{eq:kpart}
\begin{split}
\gamma_{N,t}^{(k)} (\bx_k ; \bx'_k) = \int \rd \bx_{N-k} \;
\psi_{N,t} (\bx_k , \bx_{N-k})\overline{\psi}_{N,t}  (\bx'_k,
\bx_{N-k})\,.
\end{split}
\end{equation}
Here and in the rest of the paper we use the notation $\bx= (x_1,
x_2,\dots,x_N)$, $\bx_k = (x_1, x_2, \dots , x_k)$, $\bx'_k = (x'_1
, x'_2, \dots , x'_k)$, and $\bx_{N-k} = (x_{k+1}, x_{k+2}, \dots,
x_N)$.

\bigskip

We consider initial wave functions $\psi_{N,0}$ exhibiting complete
condensation in a one-particle state $\ph$. Thus at time $t=0$, we
assume that \begin{equation}\label{eq:gm1} \gamma^{(1)}_{N,0} \to
|\ph \rangle \langle \ph| \qquad \text{as } N \to \infty \,
.\end{equation} It turns out that the last equation immediately
implies that
\begin{equation}\label{eq:gmk} \gamma^{(k)}_{N,0} \to |\ph \rangle
\langle \ph|^{\otimes k} \qquad \text{as } N \to \infty
\end{equation} for every fixed $k \in \bN$
(the argument, due to Lieb and Seiringer, can be found in \cite{LS},
after Theorem~1). It is also interesting to notice that the
convergence (\ref{eq:gm1}) (and (\ref{eq:gmk})) in the trace class
norm is equivalent to the convergence in the weak* topology defined
on the space of trace class operators on $\bR^{3}$ (or $\bR^{3k}$,
for (\ref{eq:gmk})); we thank A. Michelangeli for pointing out this
fact to us (the proof is based on general arguments, such as
Gr\"umm's Convergence Theorem).

\medskip

Starting from the Schr\"odinger equation (\ref{eq:schr}) for the
wave function $\psi_{N,t}$, we can derive evolution equations for
the marginal densities $\gamma_{N,t}^{(k)}$. The dynamics of the
marginals is governed by a hierarchy of $N$ coupled equations
usually known as the BBGKY hierarchy.
\begin{equation}\label{eq:BBGKY}
\begin{split}
i\partial_t \gamma^{(k)}_{N,t} = \; &\sum_{j=1}^N \, \left[
-\Delta_j , \gamma^{(k)}_{N,t} \right] + \sum_{i<j}^k \, \left[ V_N
(x_i -x_j), \gamma^{(k)}_{N,t} \right] \\ &+ (N-k) \sum_{j=1}^k
\tr_{k+1} \; \left[ V_N (x_j -x_{k+1}), \gamma^{(k+1)}_{N,t}
\right]\,.
\end{split}
\end{equation}
Here $\tr_{k+1}$ denotes the partial trace over the $(k+1)$-th
particle.

\medskip

Next we study the limit $N \to \infty$ of the density
$\gamma^{(k)}_{N,t}$ for fixed $k \in \bN$. For simplicity we fix
$k=1$. {F}rom (\ref{eq:BBGKY}), the evolution equation for the
one-particle density matrix, written in terms of its kernel
$\gamma^{(1)}_{N,t} (x_1; x'_1)$ is given by
\begin{equation}\label{eq:BBGKY-1}
\begin{split}
i\partial_t \gamma^{(1)}_{N,t} (x_1, x'_1) = \; &\left(-\Delta_1 +
\Delta'_1 \right) \gamma^{(1)}_{N,t} (x_1 ; x'_1) \\ &+ (N-1) \int
\rd x_2 \, \left( V_N (x_1 -x_2) - V_N (x'_1 -x_2) \right)
\gamma^{(2)}_{N,t} (x_1, x_2;x'_1,x_2) \, .\end{split}
\end{equation}
Suppose now that $\gamma_{\infty,t}^{(1)}$ and
$\gamma^{(2)}_{\infty,t}$ are limit points (with respect to the
weak* topology) of $\gamma_{N,t}^{(1)}$ and, respectively,
$\gamma^{(2)}_{N,t}$ as $N \to \infty$. Since, formally,
\[ (N-1) V_N (x) = (N-1) N^2 V(Nx) \simeq N^3 V (Nx) \to b_0 \delta
(x) \qquad \text{with } b_0 = \int \rd x \, V(x)  \] as $N \to
\infty$, we could naively expect the limit points
$\gamma^{(1)}_{\infty,t}$ and $\gamma^{(2)}_{\infty,t}$ to satisfy
the limiting equation
\begin{equation}\label{eq:wrong}
i\partial_t \gamma^{(1)}_{\infty,t} (x_1;x'_1)= \left( -\Delta_1 +
\Delta'_1 \right) \gamma^{(1)}_{\infty,t} (x_1; x'_1) + b_0 \int \rd
x_2 \, \left( \delta (x_1 -x_2) - \delta (x'_1 -x_2) \right)
\gamma^{(2)}_{\infty,t} (x_1,x_2 ; x'_1,x_2)\,.
\end{equation}
{F}rom (\ref{eq:gmk}) we have, at time $t=0$, \begin{equation}
\begin{split}
\gamma^{(1)}_{\infty,0}  (x_1 ; x'_1) &= \ph (x_1) \overline{\ph}
(x'_1) \\
\gamma^{(2)}_{\infty,0}  (x_1, x_2 ; x'_1, x'_2) &= \ph (x_1) \ph
(x_2) \overline{\ph} (x'_1) \overline{\ph} (x'_2)\,.
\end{split}
\end{equation}
If condensation is really preserved by the time evolution, also at
time $t \neq 0$ we have
\begin{equation}\label{eq:condt}
\begin{split}
\gamma^{(1)}_{\infty,t} (x_1 ; x'_1) &= \ph_t (x_1) \overline{\ph}_t
(x'_1) \\
\gamma^{(2)}_{\infty,t} (x_1,x_2; x'_1,x'_2) &= \ph_t (x_1) \ph_t
(x_2) \overline{\ph}_t (x'_1)\overline{\ph}_t (x'_2)\,.
\end{split}
\end{equation}
Inserting (\ref{eq:condt}) in (\ref{eq:wrong}), we obtain the
self-consistent equation
\begin{equation}\label{eq:b0}
i\partial_t \ph_t = -\Delta \ph_t + b_0 |\ph_t|^2 \ph_t
\end{equation}
for the condensate wave function $\ph_t$. This equation has the same
form as the time-dependent Gross-Pitaevskii equation (\ref{eq:GP1}),
but a different coefficient in front of the nonlinearity ($b_0$
instead of $8\pi a_0$).

\medskip

The reason why we obtain the wrong coupling constant in
(\ref{eq:b0}) is that going from (\ref{eq:BBGKY-1}) to
(\ref{eq:wrong}), we took the two limits
\begin{equation}\label{eq:lims} (N-1) V_N (x) \to b_0 \delta (x)
\quad \text{ and } \quad  \gamma^{(2)}_{N,t} \to
\gamma^{(2)}_{\infty,t} \end{equation} independently from each
other. However, since the scattering length of the interaction is of
the order $1/N$, the two-particle density $\gamma^{(2)}_{N,t}$
develops a short scale correlation structure on the length scale
$1/N$, which is exactly the same length scale on which the potential
$V_N$ varies. For this reason the two limits in (\ref{eq:lims})
cannot be taken independently. In order to obtain the correct
Gross-Pitaevskii equation (\ref{eq:GP1}) we need to take into
account the correlations among the particles, and the short scale
structure they create in the marginal density $\gamma_{N,t}^{(2)}$.

\medskip

To describe the correlations among the particles we make use of the
solution $f_N (x)$ to the zero energy scattering equation
(\ref{eq:0en-N}). Assuming that the function $f_N (x_i -x_j)$ gives
a good approximation for the correlations between particles $i$ and
$j$, we may expect that the one- and two-particle densities
associated with the evolution of a condensate are given, for large
but finite $N$, by
\begin{equation}\label{eq:gam12}
\begin{split}
\gamma^{(1)}_{N,t} (x_1; x'_1) &\simeq \ph_t (x_1) \overline{\ph}_t (x'_1)\\
\gamma^{(2)}_{N,t} (x_1, x_2; x'_1, x'_2) &\simeq f_N (x_1 -x_2) f_N
(x'_1 - x'_2) \ph_t (x_1) \ph_t (x_2) \overline{\ph}_t (x'_1)
\overline{\ph}_t (x'_2)\,.
\end{split}
\end{equation}
Inserting this ansatz into (\ref{eq:BBGKY-1}), we obtain a new
self-consistent equation \begin{equation}\label{eq:GP2}
\begin{split} i\partial_t \ph_t &= -\Delta \ph_t + \left(
\lim_{N\to\infty} (N-1) \int \rd x f_N (x) V_N (x) \right) |\ph_t|^2
\ph_t \\ &= -\Delta \ph_t + \left( \lim_{N\to\infty} N^3 \int \rd x
f (Nx) V (Nx) \right) |\ph_t|^2 \ph_t \\ &= -\Delta \ph_t + 8\pi a_0
|\ph_t|^2 \ph_t
\end{split}
\end{equation}
because of (\ref{eq:8pia0}). This is exactly the Gross-Pitaevskii
equation (\ref{eq:GP1}), with the correct coupling constant in front
of the nonlinearity.

\medskip

Note that the presence of the correlation functions $f_N (x_1 -x_2)$
and $f_N (x'_1 -x'_2)$ in (\ref{eq:gam12}) does not contradict
complete condensation of the system at time $t$. On the contrary, in
the weak limit $N \to \infty$, the function $f_N$ converges to one,
and therefore $\gamma^{(1)}_{N,t}$ and $\gamma^{(2)}_{N,t}$ converge
to $|\ph_t \rangle \langle \ph_t|$ and $|\ph_t \rangle \langle
\ph_t|^{\otimes 2}$, respectively. The correlations described by the
function $f_N$ can only produce nontrivial effects on the
macroscopic dynamics of the system because of the singularity of the
interaction potential $V_N$.

\medskip

{F}rom this heuristic argument it is clear that, in order to obtain
a rigorous derivation of the Gross-Pitaevskii equation
(\ref{eq:GP2}), we need to identify the short scale structure of the
marginal densities and prove that, in a very good approximation, it
can be described by the function $f_N$ as in (\ref{eq:gam12}). In
other words, we need to show a very strong separation of scales in
the marginal density $\gamma^{(2)}_{N,t}$ (and, more generally, in
the $k$-particle density $\gamma^{(k)}_{N,t}$) associated with the
solution of the $N$-body Schr\"odinger equation; the
Gross-Pitaevskii theory can only be correct if $\gamma^{(k)}_{N,t}$
has a regular part, which factorizes for large $N$ into the product
of $k$ copies of the orthogonal projection $|\ph_t \rangle \langle
\ph_t|$, and a time independent singular part, due to the
correlations among the particles, and described by products of the
functions $f_N (x_i -x_j)$, $1 \leq i,j \leq k$.

\section{Main Results}\label{sec:main}
\setcounter{equation}{0}

To prove our main results we need to assume the interaction
potential to be sufficiently weak. To measure the strength of the
potential, we introduce the dimensionless quantity
\begin{equation}\label{eq:alpha}
\alpha = \sup_{x \in \bR^3} |x|^2 V(x) + \int \frac{\rd x}{|x|} \,
V(x) \,.
\end{equation}
Apart from the smallness assumption on the potential, we also need
to assume that the correlations characterizing the initial
$N$-particle wave function are sufficiently weak. We define
therefore the notion of {\it asymptotically factorized} wave
functions. We say that a family of permutation symmetric wave
functions $\psi_N$ is asymptotically factorized if there exists $\ph
\in L^2 (\bR^3)$ and, for any fixed $k \geq 1$, there exists a
family $\xi^{(N-k)}_N \in L_s^2 (\bR^{3(N-k)})$ such that
\begin{equation}\label{eq:asy} \left\| \, \psi_N - \ph^{\otimes k} \otimes
\xi^{(N-k)}_N \right\| \to 0 \qquad \text{as } N \to \infty \, .
\end{equation} It is simple to check that, if $\psi_N$ is
asymptotically factorized, then it exhibits complete Bose-Einstein
condensation in the one-particle state $\ph$ (in the sense that the
one-particle density associated with $\psi_N$ satisfy
$\gamma^{(1)}_N \to |\ph \rangle \langle \ph|$ as $N \to \infty$).
Asymptotic factorization is therefore a stronger condition than
complete condensation, and it provides more control on the
correlations of $\psi_N$.

\begin{theorem}\label{thm:main}
Assume that $V(x)$ is a positive, smooth, spherical symmetric, and
compactly supported potential such that $\alpha$ (defined in
(\ref{eq:alpha})) is sufficiently small. Consider an asymptotically
factorized family of wave functions $\psi_N \in L^2_s (\bR^{3N})$,
exhibiting complete Bose-Einstein condensation in a one-particle
state $\ph \in H^1 (\bR^3)$, in the sense that
\begin{equation}\label{eq:convt0} \gamma^{(1)}_{N} \to |\ph \rangle
\langle \ph| \qquad \text{as } N\to \infty
\end{equation} where $\gamma^{(1)}_N$ denotes the one-particle
density associated with $\psi_N$. Then, for any fixed $t \in \bR$,
the one-particle density $\gamma^{(1)}_{N,t}$ associated with the
solution $\psi_{N,t}$ of the $N$-particle Schr\"odinger equation
(\ref{eq:schr}) satisfies
\begin{equation}\label{eq:converg} \gamma^{(1)}_{N,t} \to |\ph_t \rangle \langle \ph_t| \qquad
\text{as } N \to \infty \end{equation} where $\ph_t$ is the solution
to the time-dependent Gross-Pitaevskii equation
\begin{equation}\label{eq:GP} i\partial_t \ph_t = -\Delta \ph_t + 8
\pi a_0 |\ph_t|^2 \ph_t \end{equation} with initial data $\ph_{t=0}
= \ph$.
\end{theorem}
The convergence in (\ref{eq:convt0}) and (\ref{eq:converg}) is in
the trace norm topology (which in this case is equivalent to the
weak* topology defined on the space of trace class operators on
$\bR^3$). Moreover, from (\ref{eq:converg}) we also get convergence
of higher marginal. For every $k \geq 1$, we have
\[ \gamma^{(k)}_{N,t} \to |\ph_t \rangle \langle \ph_t|^{\otimes k}
\, \quad \text{as } N \to \infty. \]

\bigskip

Theorem \ref{thm:main} can be used to describe the dynamics of
condensates satisfying the condition of asymptotic factorization.
The following two corollaries provide examples of such initial data.

\medskip

The simplest example of $N$-particle wave function satisfying the
assumption of asymptotic factorization is given by a product state.

\begin{corollary}\label{cor:product}
Under the assumptions on $V(x)$ stated in Theorem \ref{thm:main},
let $\psi_N (\bx) = \prod_{j=1}^N \ph (x_j)$ for an arbitrary $\ph
\in H^1 (\bR^3)$. Then, for any $t \in \bR$,
\[ \gamma^{(1)}_{N,t} \to |\ph_t \rangle \langle \ph_t| \quad
\text{as } N \to \infty \] where $\ph_t$ is a solution of the
Gross-Pitaevskii equation (\ref{eq:GP}) with initial data $\ph_{t=0}
= \ph$.
\end{corollary}

The second application of Theorem \ref{thm:main} gives a
mathematical description of the results of the experiments depicted
in the introduction.

Let \begin{equation}\label{eq:trapped} H_N^{\text{trap}} =
\sum_{j=1}^N \left(-\Delta_j + V_{\text{ext}} (x_j) \right) +
\sum_{i<j}^N V_N (x_i -x_j) \end{equation} with a confining
potential $V_{\text{ext}}$. Let $\psi_N$ be the ground state of
$H_N^{\text{trap}}$. By \cite{LS}, $\psi_N$ exhibits complete Bose
Einstein condensation into the minimizer $\phi_{\text{GP}}$ of the
Gross-Pitaevskii energy functional $\cE_{\text{GP}}$ defined in
(\ref{eq:GPfunc}). In other words
\[ \gamma^{(1)}_N \to |\phi_{\text{GP}} \rangle \langle
\phi_{\text{GP}} | \qquad \text{as  } N \to \infty \, .\] In
\cite{ESY2}, we demonstrate that $\psi_N$ also satisfies the
condition (\ref{eq:asy}) of asymptotic factorization. {F}rom this
observation, we obtain the following corollary.

\begin{corollary}\label{cor:trap}
Under the assumptions on $V(x)$ stated in Theorem \ref{thm:main},
let $\psi_N$ be the ground state of (\ref{eq:trapped}), and denote
by $\gamma_{N,t}^{(1)}$ the one-particle density associated with the
solution $\psi_{N,t}= e^{-iH_N t} \psi_N$ of the Schr\"odinger
equation (\ref{eq:schr}). Then, for any fixed $t \in \bR$,
\[ \gamma^{(1)}_{N,t} \to |\ph_t \rangle \langle \ph_t| \qquad \text{as } N \to \infty \]
where $\ph_t$ is the solution of the Gross-Pitaevskii equation
(\ref{eq:GP}) with initial data $\ph_{t=0} = \phi_{\text{GP}}$.
\end{corollary}

Although the second corollary describes physically more realistic
situations, also the first corollary has interesting consequences.
In Section \ref{sec:heu}, we observed that the emergence of the
scattering length in the Gross-Pitaevskii equation is an effect due
to the correlations. The fact that the Gross-Pitaevskii equation
describes the dynamics of the condensate also if the initial wave
function is completely uncorrelated, as in Corollary
\ref{cor:product}, implies that the $N$-body Schr\"odinger dynamics
generates the singular correlation structure in very short times. Of
course, when the wave function develops correlations on the length
scale $1/N$, the energy associated with this length scale decreases;
since the total energy is conserved by the Schr\"odinger evolution,
we must conclude that together with the short scale structure at
scales of order $1/N$, the $N$-body dynamics also produces
oscillations on intermediate length scales $1/N \ll \ell \ll 1$,
which carry the excess energy (the difference between the energy of
the factorized wave function and the energy of the wave function
with correlations on the length scale $1/N$) and which have no
effect on the macroscopic dynamics (because only variations of the
wave function on length scales of order one and order $1/N$ affect
the macroscopic dynamics described by the Gross-Pitaevskii
equation).

\section{General Strategy of the Proof and Previous Results}
\label{sec:proof}\setcounter{equation}{0}

In this section we illustrate the strategy used to prove Theorem
\ref{thm:main}. The proof is divided into three main steps.

\medskip

{\it Step 1. Compactness of $\gamma_{N,t}^{(k)}$.} Recall, from
(\ref{eq:kpart}), the definition of the marginal densities
$\gamma^{(k)}_{N,t}$ associated with the solution $\psi_{N,t} = \exp
(-iH_N t) \psi_N$ of the $N$-body Schr\"odinger equation. By
definition, for any $N\in \bN$ and $t\in \bR$, $\gamma^{(k)}_{N,t}$
is a positive operator in $\cL^1_k = \cL^1 (L^2(\bR^{3k}))$ (the
space of trace class operators on $L^2 (\bR^{3k})$) with trace equal
to one. For fixed $t \in \bR$ and $k \geq 1$, it follows by standard
general argument (Banach-Alaouglu Theorem) that the sequence
$\{\gamma^{(k)}_{N,t} \}_{N \geq k}$ is compact with respect to the
weak* topology of $\cL^1_k$. Note here that $\cL^1_k$ has a weak*
topology because $\cL^1_k = \cK_k^*$, where $\cK_k = \cK (L^2
(\bR^{3k}))$ is the space of compact operators on $L^2 (\bR^{3k})$.
To make sure that we can find subsequences of $\gamma_{N,t}^{(k)}$
which converge for all times in a certain interval, we fix $T > 0$
and consider the space $C([0,T] , \cL^1_k)$ of all functions of $t
\in [0,T]$ with values in $\cL_k^1$ which are continuous with
respect to the weak* topology on $\cL_k^1$. Since $\cK_k$ is
separable, it follows that the weak* topology on the unit ball of
$\cL^1_k$ is metrizable; this allows us to prove the equicontinuity
of the densities $\gamma_{N,t}^{(k)}$, and to obtain compactness of
the sequences $\{\gamma^{(k)}_{N,t}\}_{N\geq k}$ in $C([0,T],
\cL_k^1)$.

\bigskip

{\it Step 2. Convergence to an infinite hierarchy.} By Step 1 we
know that, as $N \to \infty$, the family of marginal densities
$\Gamma_{N,t} = \{ \gamma_{N,t}^{(k)} \}_{k=1}^N$ has at least one
limit point $\Gamma_{\infty,t} = \{ \gamma^{(k)}_{\infty,t} \}_{k
\geq 1}$ in $\bigoplus_{k \geq 1} C([0,T], \cL^1_k)$ with respect to
the product topology. Next, we derive evolution equations for the
limiting densities $\gamma^{(k)}_{\infty,t}$. Starting from the
BBGKY hierarchy (\ref{eq:BBGKY}) for the family $\Gamma_{N,t}$, we
prove that any limit point $\Gamma_{\infty,t}$ satisfies the
infinite hierarchy of equations
\begin{equation}\label{eq:GPhier}
i \partial_t \gamma^{(k)}_{\infty,t} = \sum_{j=1}^k \left[ -\Delta_j
, \gamma^{(k)}_{\infty,t} \right] + 8 \pi a_0 \sum_{j=1}^k \tr_{k+1}
\; \left[ \delta (x_j -x_{k+1}), \gamma^{(k+1)}_{\infty,t} \right]\,
\end{equation}
for $k \geq 1$. It is at this point, in the derivation of this
infinite hierarchy, that we need to identify the singular part of
the densities $\gamma_{N,t}^{(k+1)}$. The emergence of the
scattering length in the second term on the right hand side of
(\ref{eq:GPhier}) is due to short scale structure of
$\gamma_{N,t}^{(k+1)}$.

\medskip

It is worth noticing that the infinite hierarchy (\ref{eq:GPhier})
has a factorized solution. In fact, it is simple to see that the
infinite family \begin{equation}\label{eq:fact}
 \gamma^{(k)}_t =
|\ph_t \rangle \langle \ph_t|^{\otimes k} \qquad \text{for } k \geq
1
\end{equation} solves (\ref{eq:GPhier}) if and only if $\ph_t$ is
a solution to the Gross-Pitaevskii equation (\ref{eq:GP}).

\bigskip

{\it Step 3. Uniqueness of the solution to the infinite hierarchy.}
To conclude the proof of Theorem~\ref{thm:main}, we show that the
infinite hierarchy (\ref{eq:GPhier}) has a unique solution. This
implies immediately that the densities $\gamma_{N,t}^{(k)}$
converge; in fact, a compact sequence with at most one limit point
is always convergent. Moreover, since we know that the factorized
densities (\ref{eq:fact}) are a solution, it also follows that, for
any $k \geq 1$, \[ \gamma^{(k)}_{N,t} \to |\ph_t \rangle \langle
\ph_t|^{\otimes k} \qquad \text{as } N \to \infty \] with respect to
the weak* topology of $\cL^1_k$.

\bigskip

Similar strategies have been used to obtain rigorous derivations of
the nonlinear Hartree equation \begin{equation}\label{eq:hartree}
i\partial_t \ph_t = -\Delta \ph_t + (v * |\ph_t|^2 ) \ph_t
\end{equation} for the dynamics of initially factorized wave functions
in bosonic many particle mean field models, characterized by the
Hamiltonian \begin{equation}\label{eq:mf} H^{\text{mf}}_N =
\sum_{j=1}^N -\Delta_j + \frac{1}{N} \sum_{i<j}^N  v (x_i -x_j) \,.
\end{equation}

\medskip

In this context, the approach outlined above was introduced by Spohn
in \cite{Sp}, who applied it to derive (\ref{eq:hartree}) in the
case of a bounded potential $v$. In \cite{EY}, Erd\H os and Yau
extended Spohn's result to the case of a Coulomb interaction $v(x) =
\pm 1 / |x|$ (partial results for the Coulomb case, in particular
the convergence to the infinite hierarchy, were also obtained by
Bardos, Golse, and Mauser, see \cite{BGM}). More recently, Adami,
Golse, and Teta used the same approach in \cite{AGT} for
one-dimensional systems with dynamics generated by a Hamiltonian of
the form (\ref{eq:mf}) with an $N$-dependent pair potential $v_N (x)
= N^{\beta} V(N^{\beta} x)$, $\beta <1$. In the limit $N \to
\infty$, they obtain the nonlinear Schr\"odinger equation \[
i\partial_t \ph_t = -\Delta \ph_t + b_0 |\ph_t|^2 \ph_t \qquad
\text{with } b_0 = \int V(x) \rd x  \, .\]

\medskip

Notice that the Hamiltonian (\ref{eq:ham}) has the same form as the
mean field Hamiltonian (\ref{eq:mf}), with an $N$-dependent pair
potential $v_N (x) = N^3 V (Nx)$. Of course, one may also ask what
happens if we consider the mean field Hamiltonian (\ref{eq:mf}) with
the $N$-dependent potential $v_N (x)= N^{3\beta} V (N^{\beta} x)$,
for $\beta \neq 1$. If $\beta <1$, the short scale structure
developed by the solution of the Schr\"odinger equation is still
characterized by the length scale $1/N$ (because the scattering
length of $N^{3\beta-1} V(N^{\beta} x)$ is still of order $1/N$);
but this time the potential varies on much larger scales, of the
order $N^{-\beta} \gg N^{-1}$. For this reason, if $\beta<1$, the
scattering length does not appear in the effective macroscopic
equation ($8\pi a_0$ is replaced by $b_0 = \int \rd x \, V(x)$). In
\cite{ESY3} (and previously in \cite{ESY2} for $0< \beta <1/2$) we
prove in fact that Corollary \ref{cor:product} can be extended to
include the case $0< \beta <1$ as follows.
\begin{theorem}
Suppose $\psi_N (\bx) = \prod_{j=1}^N \ph (x_j)$, for some $\ph \in
H^1 (\bR^3)$. Let $\psi_{N,t}= e^{-iH_{\beta,N} t} \psi_N$ with the
mean-field Hamiltonian
\[ H_{\beta,N} = \sum_{j=1}^N -\Delta_j + \frac{1}{N} \sum_{i<j}^N N^{3\beta}
V(N^{\beta}(x_i -x_j)) \] for a positive, spherical symmetric,
compactly supported, and smooth potential $V$ such that $\a$
(defined in (\ref{eq:alpha})) is sufficiently small. Let
$\gamma^{(1)}_{N,t}$ be the one-particle density associated with
$\psi_{N,t}$. Then, if $\; 0< \beta \leq 1$ we have, for any fixed
$t \in \bR$, \( \gamma^{(1)}_{N,t} \to |\ph_t \rangle \langle
\ph_t|\) as $N \to \infty$. Here $\ph_t$ is the solution to the
nonlinear Schr\"odinger equation
\[ i\partial_t \ph_t = -\Delta \ph_t + \sigma |\ph_t|^2 \ph_t \]
with initial data $\ph_{t=0} = \ph$ and with \[ \sigma = \left\{
\begin{array}{ll} 8 \pi a_0 \quad & \text{if } \beta =1 \\
b_0  \quad & \text{if } 0 < \beta < 1
\end{array} \right. \, .\]
\end{theorem}


\section{Convergence to the Infinite Hierarchy}\label{sec:conv}
\setcounter{equation}{0}

In this section we give some more details concerning Step 2 in the
strategy outlined above. We consider a limit point
$\Gamma_{\infty,t} = \{ \gamma^{(k)}_{\infty,t} \}_{k \geq 1}$ of
the sequence $\Gamma_{N,t} = \{ \gamma^{(k)}_{N,t} \}_{k=1}^N$ and
we prove that $\Gamma_{\infty,t}$ satisfies the infinite hierarchy
(\ref{eq:GPhier}). To this end we use that, for finite $N$, the
family $\Gamma_{N,t}$ satisfies the BBGKY hierarchy
(\ref{eq:BBGKY}), and we show the convergence of each term in
(\ref{eq:BBGKY}) to the corresponding term in the infinite hierarchy
(\ref{eq:GPhier}) (the second term on the r.h.s. of (\ref{eq:BBGKY})
is of smaller order and can be proven to vanish in the limit $N \to
\infty$).

\medskip

The main difficulty consists in proving the convergence of the last
term on the right hand side of (\ref{eq:BBGKY}) to the last term on
the right hand side of (\ref{eq:GPhier}). In particular, we need to
show that in the limit $N \to \infty$ we can replace the potential
$(N-k) N^2 V (N (x_j -x_{k+1})) \simeq N^3 V(Nx)$ in the last term
on the r.h.s. of (\ref{eq:BBGKY}) by $8 \pi a_0 \delta (x_j -
x_{k+1})$ . In terms of kernels we have to prove that
\begin{equation}\label{eq:co} \int \rd x_{k+1} \, \left( N^3 V (N (x_j -
x_{k+1})) - 8 \pi a_0 \delta (x_j -x_{k+1}) \right) \,
\gamma^{(k+1)}_{N,t} (\bx_k , x_{k+1}, \bx'_k , x_{k+1} ) \to 0
\end{equation} as $N \to \infty$. It is enough to prove the convergence
(\ref{eq:co}) in a weak sense, after testing the expression against
a smooth $k$-particle kernel $J^{(k)} (\bx_k ; \bx'_k)$. Note,
however, that the observable $J^{(k)}$ does not help to perform the
integration over the variable $x_{k+1}$.

\medskip

The problem here is that, formally, the $N$-dependent potential $N^3
V(N (x_j -x_{k+1}))$ does not converge towards $8 \pi a_0 \delta
(x_j - x_{k+1})$ as $N \to \infty$ (it converges towards $b_0 \delta
(x_j -x_{k+1})$, with $b_0 = \int \rd x \, V(x)$). Eq. (\ref{eq:co})
is only correct because of the correlations between $x_j$ and
$x_{k+1}$ hidden in the density $\gamma^{(k+1)}_{N,t}$. Therefore,
to prove (\ref{eq:co}), we start by factoring out the correlations
explicitly, and by proving that, as $N \to \infty$,
\begin{equation}\label{eq:co2} \int \rd x_{k+1} \, \left( N^3 V (N (x_j -
x_{k+1})) f_N (x_j -x_{k+1}) - 8 \pi a_0 \delta (x_j -x_{k+1})
\right) \,  \frac{\gamma^{(k+1)}_{N,t} (\bx_k , x_{k+1}, \bx'_k ,
x_{k+1} )}{f_N (x_j -x_{k+1})} \to 0,
\end{equation}
where $f_N (x)$ is the solution to the zero energy scattering
equation (\ref{eq:0en-N}). Then, in a second step, we use the fact
that $f_N \to 1$ in the weak limit $N \to \infty$, to prove that the
ratio $\gamma^{(k+1)}_{N,t}/ f_N (x_j -x_{k+1})$ converges to the
same limiting density $\gamma_{\infty,t}^{(k+1)}$ as
$\gamma^{(k+1)}_{N,t}$. Eq. (\ref{eq:co2}) looks now much better
than (\ref{eq:co}) because, formally, $N^3 V (N (x_j -x_{k+1})) f_N
(x_j -x_{k+1})$ does converge to $8\pi a_0 \delta (x_j - x_{k+1})$.
To prove that (\ref{eq:co2}) is indeed correct, we only need some
regularity of the ratio $\gamma^{(k+1)}_{N,t} (\bx_k , x_{k+1};
\bx'_k ,x_{k+1})/ f_N (x_j -x_{k+1})$ in the variables $x_j$ and
$x_{k+1}$. In terms of the $N$-particle wave function $\psi_{N,t}$
we need regularity of $\psi_{N,t} (\bx)/f_N (x_i - x_j)$ in the
variables $x_i$, $x_j$, for any $i \neq j$. To establish the
required regularity we use the following energy estimate.

\begin{proposition}\label{prop:enest}
Consider the Hamiltonian $H_N$ defined in (\ref{eq:ham}), with a
positive, spherical symmetric, smooth and compactly supported
potential $V$. Suppose that $\a$ (defined in (\ref{eq:alpha})) is
sufficiently small. Then there exists $C = C(\alpha)
>0$ such that
\begin{equation}\label{eq:enesti}
\langle \psi, H_N^2 \psi \rangle \geq C N^2 \, \int \rd \bx \;
\left| \nabla_{i} \nabla_{j} \frac{\psi (\bx)}{f_N (x_i -x_j)}
\right|^2 \, .
\end{equation}
for all $i \neq j$ and for all $\psi \in L^2_s (\bR^{3N}, \rd \bx)$.
\end{proposition}

Making use of this energy estimate it is possible to deduce strong
a-priori bounds on the solution $\psi_{N,t}$ of the Schr\"odinger
equation (\ref{eq:schr}). These bounds have the form
\begin{equation}\label{eq:aprior}
\int \rd \bx \, \left| \nabla_{i} \nabla_{j} \frac{\psi_{N,t}
(\bx)}{f_N (x_i -x_j)} \right|^2 \leq C
\end{equation}
uniformly in $N \in \bN$ and $t \in \bR$. To prove (\ref{eq:aprior})
we use that, by (\ref{eq:enesti}), and because of the conservation
of the energy along the time evolution,
\begin{equation}\label{eq:apri2} \int \rd \bx \, \left| \nabla_{i} \nabla_{j}
\frac{\psi_{N,t} (\bx)}{f_N (x_i -x_j)} \right|^2 \leq C N^{-2}
\langle \psi_{N,t}, H_N^2 \psi_{N,t} \rangle = C N^{-2} \langle
\psi_{N,0}, H_N^2 \psi_{N,0} \rangle \, .\end{equation} {F}rom
(\ref{eq:apri2}) and using an approximation argument on the initial
wave function to make sure that the expectation of $H^2_N$ at time
$t=0$ is of the order $N^2$, we obtain (\ref{eq:aprior}).

\medskip

The bounds (\ref{eq:aprior}) are then sufficient to prove the
convergence (\ref{eq:co}) (using a non-standard Poincar\'e
inequality; see Lemma 7.2 in \cite{ESY3}).

\medskip

Remark that the a-priori bounds (\ref{eq:aprior}) do not hold true
if we do not divide the solution $\psi_{N,t}$ of the Schr\"odinger
equation by $f_N (x_i -x_j)$ (replacing $\psi_{N,t} (\bx)/ f_N (x_i
-x_j)$ by $\psi_N (\bx)$ the integral in (\ref{eq:aprior}) would be
of order $N$). It is only after removing the singular factor $f_N
(x_i -x_j)$ from $\psi_{N,t} (\bx)$ that we can prove useful bounds
on the regular part of the wave function.

\medskip

It is through the a-priori bounds (\ref{eq:aprior}) that we identify
the correlation structure of the wave function $\psi_{N,t}$ and that
we show that, when $x_i$ and $x_j$ are close to each other,
$\psi_{N,t} (\bx)$ can be approximated by the time independent
singular factor $f_N (x_i -x_j)$, which varies on the length scale
$1/N$, multiplied with a regular part (regular in the sense that it
satisfy the bounds (\ref{eq:aprior})). It is therefore through
(\ref{eq:aprior}) that we establish the strong separation of scales
in the wave function $\psi_{N,t}$ and in the marginal densities
$\gamma^{(k)}_{N,t}$ which is of fundamental importance for the
Gross-Pitaevskii theory.

\medskip

Since it is quite short and it shows why the solution $f_N (x_i
-x_j)$ to the zero energy scattering equation (\ref{eq:0en}) can be
used to describe the two-particle correlations, we reproduce in the
following the proof Proposition \ref{prop:enest}. Note that this is
the only step in the proof of our main theorem where the smallness
of constant $\a$, measuring the strength of the interaction
potential, is used. The positivity of the interaction potential, on
the other hand, also plays an important role in many other parts of
the proof.

\begin{proof}[Proof of Proposition \ref{prop:enest}]
We decompose the Hamiltonian (\ref{eq:ham}) as
\[ H_N = \sum_{j=1}^N \, h_j \qquad \text{with} \qquad h_j = -\Delta_j
+ \frac{1}{2} \sum_{i \neq j} \, V_N (x_i - x_j)\,. \] For an
arbitrary permutation symmetric wave function $\psi$ and for any
fixed $i \neq j$, we have
\[ \langle \psi, H^2_N \psi \rangle = N \langle \psi, h_i^2
\psi \rangle + N (N-1) \langle \psi, h_i h_j \psi \rangle \geq N
(N-1) \langle \psi, h_i h_j \psi \rangle \,. \] Using the positivity
of the potential, we find
\begin{equation}\label{eq:pro1}\langle \psi, H^2_N \psi \rangle \geq N
(N-1) \left\langle \psi, \left(-\Delta_i + \frac{1}{2} V_N (x_i
-x_j) \right) \left(-\Delta_j + \frac{1}{2} V_N (x_i -x_j) \right)
\psi \right\rangle \,. \end{equation} Next, we define $\phi (\bx)$
by $\psi (\bx) = f_N (x_i -x_j) \, \phi (\bx)$ ($\phi$ is well
defined because $f_N (x) >0$ for all $x \in \bR^3$); note that the
definition of the function $\phi$ depends on the choice of $i,j$.
Then
\[ \frac{1}{f_N (x_i -x_j)} \Delta_i \, \left( f_N (x_i-x_j)
\phi (\bx) \right) = \Delta_i \phi (\bx) + \frac{(\Delta f_N )(x_i
-x_j)}{f_N (x_i -x_j)} \phi (\bx) + \frac{\nabla f_N (x_i -x_j)}{f_N
(x_i -x_j)} \nabla_i \phi (\bx) \, . \] {F}rom (\ref{eq:0en}) it
follows that
\[ \frac{1}{f_N (x_i -x_j)} \left( -\Delta_i + \frac{1}{2} V_N (x_i
-x_j) \right) f_N (x_i -x_j) \phi (\bx) = L_i \phi (\bx)
\] and analogously
\[
\frac{1}{f_N (x_i -x_j)} \left(-\Delta_j + \frac{1}{2} V_N (x_i
-x_j) \right) f_N (x_i -x_j) \phi (\bx) = L_j \phi (\bx)
\]
where we defined
\[ L_{\ell}  = -\Delta_{\ell} + 2 \frac{\nabla_{\ell} \, f_N (x_i
-x_j)}{f_N (x_i -x_j)} \, \nabla_{\ell}, \qquad \text{for } \quad
\ell=i,j \,.
\] Remark that, for $\ell =i,j$, the operator $L_{\ell}$ satisfies
\[ \int \rd \bx \, f_N^2 (x_i -x_j) \; L_{\ell} \, \overline{\phi} (\bx) \;
\psi (\bx) =\int \rd \bx \, f_N^2 (x_i -x_j) \; \overline{\phi}
(\bx) \; L_{\ell} \, \psi (\bx) = \int \rd \bx \, f_N^2 (x_i -x_j)
\; \nabla_{\ell} \, \overline{\phi} (\bx) \; \nabla_{\ell} \, \psi
(\bx) \, .\] Therefore, from (\ref{eq:pro1}), we obtain
\begin{equation}\label{eq:pro2}
\begin{split}
\langle \psi, H^2_N \psi \rangle \geq \; &N (N-1) \int \rd \bx \;
f^2_N (x_i -x_j) \; L_i\, \overline{\phi} (\bx)\, L_j\, \phi (\bx)
\\ = \;&N (N-1) \int \rd \bx \; f^2_N (x_i -x_j) \; \nabla_i
\overline{\phi} (\bx)\, \nabla_i L_j \, \phi (\bx)
\\= \;&N (N-1) \int \rd \bx \; f^2_N (x_i -x_j) \;
\nabla_i \overline{\phi} (\bx)\, L_j \, \nabla_i \phi (\bx)
\\ &+N (N-1) \int \rd \bx \; f^2_N (x_i -x_j) \;
\nabla_i \overline{\phi} (\bx)\, [\nabla_i, L_j] \phi (\bx)
\\= \;&N (N-1) \int \rd \bx \; f^2_N (x_i -x_j) \; \left| \nabla_j
\nabla_i \phi (\bx) \right|^2
\\ &+N (N-1) \int \rd \bx \; f^2_N (x_i -x_j) \; \left( \nabla_i
\frac{\nabla f_N (x_i -x_j)}{f_N (x_i -x_j)} \right) \;  \nabla_i
\overline{\phi} (\bx)\, \nabla_j \phi (\bx)\,.
\end{split}
\end{equation}
To control the second term on the right hand side of the last
equation we use bounds on the function $f_N$, which can be derived
from the zero energy scattering equation (\ref{eq:0en}):
\begin{equation}\label{eq:pro3} 1 - C \alpha \leq f_N (x) \leq 1, \quad |\nabla f_N
(x)| \leq C \frac{ \alpha }{|x|}, \quad |\nabla^2 f_N (x) | \leq C
\frac{\alpha}{|x|^2} \end{equation} for constants $C$ independent of
$N$ and of the potential $V$ (recall the definition of the
dimensionless constant $\alpha$ from (\ref{eq:alpha})). Therefore,
for $\alpha <1$,
\begin{equation}
\begin{split}
\Big| \int \rd \bx \; f^2_N (x_i -x_j) \; &\left( \nabla_i
\frac{\nabla f_N (x_i -x_j)}{f_N (x_i -x_j)} \right) \;  \nabla_i
\overline{\phi} (\bx)\, \nabla_j \phi (\bx) \Big| \\
\leq \; & C \alpha \int \rd \bx \; \frac{1}{|x_i -x_j|^2} \,
|\nabla_i \phi (\bx)| \, |\nabla_j \phi (\bx)| \\ \leq \; & C \alpha
\int \rd \bx \; \frac{1}{|x_i -x_j|^2} \, \left( |\nabla_i \phi
(\bx)|^2 + |\nabla_j \phi (\bx)|^2 \right)
\\ \leq \; &C \alpha \int \rd \bx \; |\nabla_i \nabla_j \phi
(\bx)|^2
\end{split}
\end{equation}
where we used Hardy inequality. Thus, from (\ref{eq:pro2}), and
using again the first bound in (\ref{eq:pro3}), we obtain
\[\langle \psi, H^2_N \psi \rangle \geq \; N (N-1) (1- C \alpha)
\int \rd \bx \left| \nabla_i \nabla_j \phi (\bx) \right|^2 \] which
implies (\ref{eq:enesti}).
\end{proof}

\section{Uniqueness of the Solution to the Infinite
Hierarchy}\label{sec:unique}\setcounter{equation}{0}

In this section we discuss the main ideas used to prove the
uniqueness of the solution to the infinite hierarchy (Step 3 in the
strategy outlined in Section \ref{sec:proof}).

\medskip

First of all, we need to specify in which class of family of
densities $\Gamma_t = \{ \gamma_{t}^{(k)} \}_{k \geq 1}$ we want to
prove the uniqueness of the solution to the infinite hierarchy
(\ref{eq:GPhier}). Clearly, the proof of the uniqueness is simpler
if we can restrict our attention to smaller classes. But of course,
in order to apply the uniqueness result to prove Theorem
\ref{thm:main}, we need to make sure that any limit point of the
sequence $\Gamma_{N,t} = \{ \gamma^{(k)}_{N,t} \}_{k=1}^N$ is in the
class for which we can prove uniqueness.

\medskip

We are going to prove uniqueness for all families $\Gamma_t = \{
\gamma^{(k)}_t \}_{k \geq 1} \in \bigoplus C([0,T], \cL_k^1)$ with
\begin{equation}\label{eq:aprk} \| \gamma_t^{(k)} \|_{\cH_k} := \tr \; \left|
(1-\Delta_1)^{1/2} \dots (1-\Delta_k)^{1/2} \, \gamma^{(k)}_t
(1-\Delta_k)^{1/2} \dots (1-\Delta_1)^{1/2} \right| \leq C^k
\end{equation} for all $t \in [0,T]$ and for all $k \geq 1$ (with a constant
$C$ independent of $k$).

\medskip

The following proposition guarantees that any limit point of the
sequence $\Gamma_{N,t}$ satisfies (\ref{eq:aprk}).

\begin{proposition} \label{prop:aprik} Assume the same conditions
as in Proposition \ref{prop:enest}. Suppose that $\Gamma_{\infty,t}
= \{\gamma_{\infty,t}^{(k)} \}_{k\geq 1}$ is a limit point of
$\Gamma_{N,t} = \{ \gamma^{(k)}_{N,t} \}_{k=1}^N$ with respect to
the product topology on $\bigoplus_{k \geq 1} C([0,T], \cL^1_k)$.
Then $\gamma^{(k)}_{\infty,t} \geq 0$ and there exists a constant
$C$ such that
\begin{equation}\label{eq:aprik3}
\tr \; (1- \Delta_1) \dots (1-\Delta_k) \gamma^{(k)}_{\infty,t} \leq
C^k
\end{equation}
for all $k \geq 1$ and $t \in [0,T]$.
\end{proposition}

Because of Proposition \ref{prop:aprik}, it is enough to prove the
uniqueness of the infinite hierarchy (\ref{eq:GPhier}) in the
following sense.

\begin{theorem}\label{thm:unique} Suppose that $\Gamma = \{ \gamma^{(k)} \}_{k \geq
1}$ is such that \begin{equation}\label{eq:aprit0} \| \gamma^{(k)}
\|_{\cH_k} \leq C^k\end{equation} for all $k \geq 1$ (the norm $\| .
\|_{\cH_k}$ is defined in (\ref{eq:aprk})). Then there exists at
most one solution $\Gamma_{t} = \{ \gamma^{(k)}_{t} \}_{k \geq 1}
\in \bigoplus C ([0,T], \cL_k)$ of (\ref{eq:GPhier}) such that
$\Gamma_{t=0} = \Gamma$ and
\begin{equation}\label{eq:apri}
\| \gamma^{(k)}_t \|_{\cH_k} \leq C^k \end{equation} for all $k \geq
1$ and all $t \in [0,T]$ (with the same constant $C$ as in
(\ref{eq:aprit0})).
\end{theorem}

In the next two subsections we explain the main ideas of the proofs
of Proposition \ref{prop:aprik} and Theorem~\ref{thm:unique}.

\subsection{Higher Order Energy Estimates}

The main difficulty in proving Proposition \ref{prop:aprik} is the
fact that the estimate (\ref{eq:aprik3}) does not hold true if we
replace $\gamma^{(k)}_{\infty,t}$ by the marginal density
$\gamma^{(k)}_{N,t}$. More precisely,
\begin{equation}\label{eq:aprwrong} \tr \; (1-\Delta_1) \dots
(1-\Delta_k) \gamma^{(k)}_{N,t} \leq C^k \end{equation} cannot hold
true with a constant $C$ independent of $N$. In fact, for finite $N$
and $k >1$, the $k$-particle density $\gamma_{N,t}^{(k)}$ still
contains the short scale structure due to the correlations among the
particles. Therefore, when we take derivatives of
$\gamma_{N,t}^{(k)}$ as in (\ref{eq:aprwrong}), the singular
structure (which varies on a length scale of order $1/N$) generates
contributions which diverge in the limit $N \to \infty$.

\medskip

To overcome this problem, we cutoff the wave function $\psi_{N,t}$
when two or more particles come at distances smaller than some
intermediate length scale $\ell$, with $N^{-1} \ll \ell \ll 1$ (more
precisely, the cutoff will be effective only when one or more
particles come close to one of the variable $x_j$ over which we want
to take derivatives). For fixed $j= 1, \dots ,N$, we define
$\theta_j \in C^{\infty} (\bR^{3N})$ such that
\[ \theta_j (\bx) \simeq \left\{ \begin{array}{ll} 1 & \quad
\text{if} \quad |x_i -x_j| \gg \ell \quad \text{for all } i \neq j
\\
0 & \quad \text{if there exists } i \neq j \quad \text{with} \quad
|x_i -x_j| \lesssim \ell \end{array} \right. \; . \] It is
important, for our analysis, that $\theta_j$ controls its
derivatives (in the sense that, for example, $|\nabla_i \theta_j|
\leq C \ell^{-1} \theta^{1/2}_j$); for this reason we cannot use
standard compactly supported cutoffs, but instead we have to
construct appropriate functions which decay exponentially when
particles come close together. Making use of the functions $\theta_j
(\bx)$, we prove the following higher order energy estimates.

\begin{proposition}\label{prop:highen}
Choose $\ell \ll 1$ such that $N \ell^2 \gg 1$. Suppose that
$\alpha$ is small enough. Then there exist constants $C_1$ and $C_2$
such that, for any $\psi \in L^2_s (\bR^{3N})$,
\begin{equation}\label{eq:highen}
\langle \psi, (H_N +C_1 N)^k \psi \rangle \geq C_2 N^k \int \rd \bx
\; \theta_1 (\bx) \dots \theta_{k-1} (\bx) \, |\nabla_1 \dots
\nabla_k \psi (\bx)|^2 \, .
\end{equation}
\end{proposition}

The meaning of the bounds (\ref{eq:highen}) is clear. We can control
the $L^2$-norm of the $k$-th derivative $\nabla_1 \dots \nabla_k
\psi$ by the expectation of the $k$-th power of the energy per
particle, if we only integrate over configurations where the first
$k-1$ particles are ``isolated'' (in the sense that there is no
particle at distances smaller than $\ell$ from $x_1, x_2, \dots,
x_{k-1}$). In this sense the energy estimate in Proposition
\ref{prop:enest} (which, compared with Proposition
\ref{prop:highen}, is restricted to $k=2$) is more precise than
(\ref{eq:highen}), because it identifies and controls the
singularity of the wave function exactly in the region cutoff from
the integral on the right side of (\ref{eq:highen}). The point is
that, while Proposition \ref{prop:enest} is used to identify the
two-particle correlations in the marginal densities
$\gamma_{N,t}^{(k)}$ (which are essential for the emergence of the
scattering length $a_0$ in the infinite hierarchy
(\ref{eq:GPhier})), we only need Proposition \ref{prop:highen} to
establish properties of the limiting densities; this is why we can
introduce cutoffs in (\ref{eq:highen}), provided we can show their
effect to vanish in the limit $N \to \infty$.

\medskip

Note that we can allow one ``free derivative''; in (\ref{eq:highen})
we take the derivative over $x_k$ although there is no cutoff
$\theta_k (\bx)$. The reason is that the correlation structure
becomes singular, in the $L^2$-sense, only when we derive it twice
(if one uses the zero energy solution $f_N$ introduced in
(\ref{eq:0en}) to describe the correlations, this can be seen by
observing that $\nabla f_N (x) \simeq 1/ |x|$, which is locally
square integrable). Remark that the condition $N\ell^2 \gg 1$ is a
consequence of the fact that, if $\ell$ is too small, the error due
to the localization of the kinetic energy on distances of order
$\ell$ cannot be controlled. The proof of Proposition
\ref{prop:highen} is based on induction over $k$; for details see
Section 9 in \cite{ESY3}.

\medskip

{F}rom the estimates (\ref{eq:highen}), using the preservation of
the expectation of $H_N^k$ along the time evolution and a
regularization of the initial $N$-particle wave function $\psi_{N}$,
we obtain the following bounds for the solution $\psi_{N,t} =
e^{-iH_N t} \psi_N$ of the Schr\"odinger equation (\ref{eq:schr}).
\begin{equation}
\int \rd \bx \; \theta_1 (\bx) \dots \theta_{k-1} (\bx) \; \left|
\nabla_1 \dots \nabla_k \psi_{N,t} (\bx) \right|^2 \leq C^k
\end{equation}
uniformly in $N$ and $t$, and for all $k \geq 1$. Translating these
bounds in the language of the density matrix $\gamma_{N,t}$, we
obtain \begin{equation}\label{eq:cutoffs} \tr \; \theta_1 \dots
\theta_{k-1} \nabla_1 \dots \nabla_k \gamma_{N,t} \nabla_1^* \dots
\nabla_k^* \leq C^k \,. \end{equation} The idea now is to use the
freedom in the choice of the cutoff length $\ell$. If we fix the
position of all particles but $x_j$, it is clear that the cutoff
$\theta_j$ is effective in a volume at most of the order $N\ell^3$.
If we choose now $\ell$ such that $N\ell^3 \to 0$ as $N \to \infty$
(which is of course compatible with the condition that $N \ell^2 \gg
1$), then we can expect that, in the limit of large $N$, the cutoff
becomes negligible. This approach yields in fact the desired
results; starting from (\ref{eq:cutoffs}), and choosing $\ell$ such
that $N \ell^3 \ll 1$, we can complete the proof of Proposition
\ref{prop:aprik} (see Proposition 6.3 in \cite{ESY3} for more
details).

\subsection{Expansion in Feynman Graphs}

To prove Theorem \ref{thm:unique}, we start by rewriting the
infinite hierarchy (\ref{eq:GPhier}) in the integral form
\begin{equation}\label{eq:GPhier-int}
\begin{split}
\gamma_{t} &= \cU^{(k)} (t) \gamma_{0} + 8i \pi a_0 \sum_{j=1}^k
\int_0^t \rd s \; \cU^{(k)} (t-s) \, \tr_{k+1} \; \left[ \delta (x_j
-x_{k+1}) , \gamma^{(k+1)}_{s} \right] \,
\\&= \cU^{(k)} (t) \gamma_{0} +
\int_0^t \rd s \; \cU^{(k)} (t-s) \, B^{(k)} \gamma^{(k+1)}_{s} \, ,
\end{split}
\end{equation}
where $\cU^{(k)} (t)$ denotes the free evolution of $k$ particles,
\[ \cU^{(k)} (t) \gamma^{(k)} = e^{it\sum_{j=1}^k \Delta_j}
\gamma^{(k)} e^{-it\sum_{j=1}^k \Delta_j} \] and the collision
operator $B^{(k)}$ maps $(k+1)$-particle operators into $k$-particle
operators according to
\begin{equation}\label{eq:B}
B^{(k)} \gamma^{(k+1)} = 8i \pi a_0  \sum_{j=1}^k \tr_{k+1}\, \left[
\delta (x_j -x_{k+1}), \gamma^{(k+1)} \right]
\end{equation}
(recall that $\tr_{k+1}$ denotes the partial trace over the
$(k+1)$-th particle).

\medskip

Iterating (\ref{eq:GPhier-int}) $n$ times we obtain the Duhamel type
series
\begin{equation}\label{eq:Duhamel}
\gamma^{(k)}_{t} = \cU^{(k)} (t) \gamma^{(k)}_0 + \sum_{m=1}^{n-1}
\xi^{(k)}_{m,t} + \eta^{(k)}_{n,t}
\end{equation}
with
\begin{equation}\label{eq:xi}
\begin{split}
\xi^{(k)}_{m,t} &= \int_0^t \rd s_1 \dots \int_0^{s_{m-1}} \rd s_m
\; \cU^{(k)} (t-s_1) B^{(k)} \cU^{(k+1)} (s_1 -s_2) B^{(k+1)} \dots
B^{(k+m-1)} \cU^{(k+m)} (s_m) \gamma^{(k+m)}_{0} \\&= \sum_{j_1=1}^k
\sum_{j_2=1}^{k+1} \dots \sum_{j_m=1}^{k+m} \int_0^t \rd s_1 \dots
\int_0^{s_{m-1}} \rd s_m \; \cU^{(k)} (t-s_1)  \tr_{k+1} \Big[
\delta (x_{j_1} - x_{k+1}),
\\  &\hspace{.3cm}  \left. \cU^{(k+1)} (s_1 -s_2) \tr_{k+2} \left[ \delta
(x_{j_2} - x_{k+2}), \dots \tr_{k+m} \left[ \delta ( x_{j_m}-
x_{k+m}), \cU^{(k+m)} (s_m) \gamma_0^{(k+m)} \right] \dots\right]
\right]
\end{split}
\end{equation} and the error term
\begin{equation}\label{eq:eta}
\eta^{(k)}_{n,t} =\int_0^t \rd s_1 \int_0^{s_1} \rd s_2 \dots
\int_0^{s_{n-1}} \rd s_n \; \cU^{(k)} (t-s_1) B^{(k)} \cU^{(k+1)}
(s_1 -s_2) B^{(k+1)} \dots B^{(k+n-1)} \gamma^{(k+m)}_{s_n}\,.
\end{equation}
Note that the error term (\ref{eq:eta}) has exactly the same form as
the terms in (\ref{eq:xi}), with the only difference that the last
free evolution is replaced by the full evolution
$\gamma^{(k+m)}_{s_n}$.

\medskip

To prove the uniqueness of the infinite hierarchy, it is enough to
prove that the error term (\ref{eq:eta}) converges to zero as $n \to
\infty$ (in some norm, or even only after testing it against a
sufficiently large class of smooth observables). The main problem
here is that the delta function in the collision operator $B^{(k)}$
cannot be controlled by the kinetic energy (in the sense that, in
three dimensions, the operator inequality $\delta (x) \leq C(1-
\Delta)$ does not hold true). For this reason, the a-priori
estimates $\| \gamma^{(k)}_t \|_{\cH_k} \leq C^k$ are not sufficient
to show that (\ref{eq:eta}) converges to zero, as $n \to \infty$.
Instead, we have to make use of the smoothing effects of the free
evolutions $\cU^{(k+j)} (s_j -s_{j+1})$ in (\ref{eq:eta}) (in a
similar way, Stricharzt estimates are used to prove the
well-posedness of nonlinear Schr\"odinger equations). To this end,
we rewrite each term in the series (\ref{eq:Duhamel}) as a sum of
contributions associated with certain Feynman graphs, and then we
prove the convergence of the Duhamel expansion by controlling each
contribution separately.

\bigskip

The details of the diagrammatic expansion can be found in Section 9
of \cite{ESY2}. Here we only present the main ideas. We start by
considering the term $\xi_{m,t}^{(k)}$ in (\ref{eq:xi}). After
multiplying it with a compact $k$-particle observable $J^{(k)}$ and
taking the trace, we expand the result as
\begin{equation}\label{eq:expa}
\tr \; J^{(k)} \xi^{(k)}_{m,t} = \sum_{\Lambda \in \cF_{m,k}}
K_{\Lambda,t}
\end{equation}
where $K_{\Lambda,t}$ is the contribution associated with the
Feynman graph $\Lambda$. Here $\cF_{m,k}$ denotes the set of all
graphs consisting of $2k$ disjoint, paired, oriented, and rooted
trees with $m$ vertices. An example of a graph in $\cF_{m,k}$ is
drawn in Figure \ref{fig:feynman}. Each vertex has one of the two
forms drawn in Figure \ref{fig:feynman}, with one ``father''-edge on
the left (closer to the root of the tree) and three ``son''-edges on
the right. One of the son edge is marked (the one drawn on the same
level as the father edge; the other two son edges are drawn below).
Graphs in $\cF_{m,k}$ have $2k+3m$ edges, $2k$ roots (the edges on
the very left), and $2k+2m$ leaves (the edges on the very right). It
is possible to show that the number of different graphs in
$\cF_{m,k}$ is bounded by $2^{4m+k}$.
\begin{figure}
\begin{center}
\epsfig{file=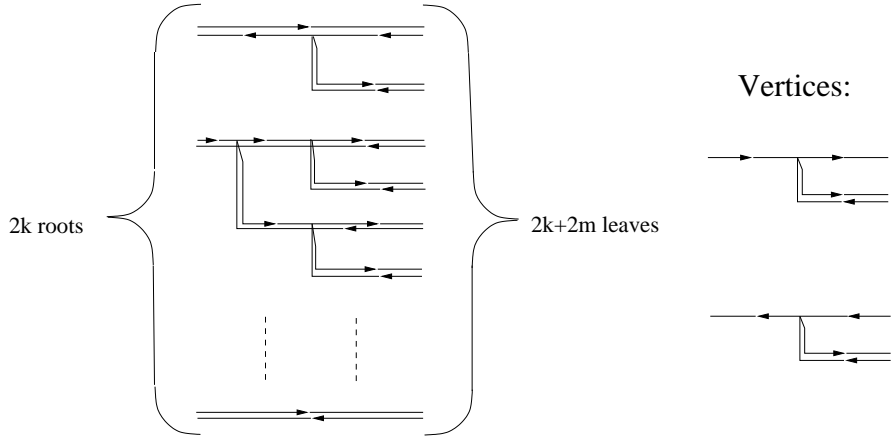,scale=.60}
\end{center}
\caption{A Feynman graph in $\cF_{m,k}$ and its two types of
vertices}\label{fig:feynman}
\end{figure}

\medskip

The particular form of the graphs in $\cF_{m,k}$ is due to the
quantum mechanical nature of the expansion; the presence of a
commutator in the collision operator (\ref{eq:B}) implies that, for
every $B^{(k+j)}$ in (\ref{eq:xi}), we can choose whether to write
the interaction on the left or on the right of the density. When we
draw the corresponding vertex in a graph in $\cF_{m,k}$, we have to
choose whether to attach it on the incoming or on the outgoing edge.

\medskip

Graphs in $\cF_{m,k}$ are characterized by a natural partial
ordering among the vertices ($v \prec v'$ if the vertex $v$ is on
the path from $v'$ to the roots); there is, however, no total
ordering. The absence of total ordering among the vertices is the
consequence of a rearrangement of the summands on the r.h.s. of
(\ref{eq:xi}); by removing the order between times associated with
non-ordered vertices we significantly reduce the number of terms in
the expansion. In fact, while (\ref{eq:xi}) contains $(m+k)!/k!$
summands, in (\ref{eq:expa}) we are only summing over $2^{4m+k}$
contributions. The price we have to pay is that the apparent gain of
a factor $1/m!$ because of the ordering of the time integrals in
(\ref{eq:xi}) is lost in the new expansion (\ref{eq:expa}). However,
since the time integrations are already needed to smooth out
singularities, and since they cannot be used simultaneously for
smoothing and for gaining a factor $1/m!$, it seems very difficult
to make use of the apparent factor $1/m!$ in (\ref{eq:xi}). In fact,
we find that the expansion (\ref{eq:expa}) is better suited for
analyzing the cumulative space-time smoothing effects of the
multiple free evolutions than (\ref{eq:xi}).

\medskip

Because of the pairing of the $2k$ trees, there is a natural pairing
between the $2k$ roots of the graph. Moreover, it is also possible
to define a natural pairing of the leaves of the graph (this is
evident in Figure \ref{fig:feynman}); two leaves $\ell_1$ and
$\ell_2$ are paired if there exists an edge $e_1$ on the path from
$\ell_1$ back to the roots, and an edge $e_2$ on the path from
$\ell_2$ to the roots, such that $e_1$ and $e_2$ are the two
unmarked son-edges of the same vertex (or, if there is no unmarked
sons in the path from $\ell_1$ and $\ell_2$ to the roots, if the two
roots connected to $\ell_1$ and $\ell_2$ are paired).

\medskip

For $\Lambda \in \cF_{m,k}$, we denote by $E(\Lambda)$,
$V(\Lambda)$, $R(\Lambda)$ and $L(\Lambda)$ the set of all edges,
vertices, roots and, respectively, leaves in the graph $\Lambda$.
For every edge $e \in E(\Lambda)$, we introduce a three-dimensional
momentum variable $p_e$ and a one-dimensional frequency variable
$\a_e$. Then, denoting by $\wh\gamma^{(k+m)}_0$ and by $\wh J^{(k)}$
the kernels of the density $\gamma^{(k+m)}_0$ and of the observable
$J^{(k)}$ in Fourier space, the contribution $K_{\Lambda,t}$ in
(\ref{eq:expa}) is given by
\begin{equation}\label{eq:contriLam}
\begin{split}
K_{\Lambda,t} = &\; \int \prod_{e \in E(\Lambda)} \frac{\rd p_e \rd
\a_e}{ \a_e - p_e^2 + i \tau_e \mu_e} \, \prod_{v \in V(\Lambda)}
\delta \left(\sum_{e \in v} \pm \a_e \right) \delta \left(\sum_{e
\in v} \pm p_e \right)  \\
&\hspace{2cm} \times  \exp \left( -it \sum_{e \in R(\Lambda)} \tau_e
(\a_e +i\tau_e \mu_e) \right) \, \wh J^{(k)} \left( \{ p_e \}_{e \in
R(\Lambda)} \right) \, \wh\gamma^{(k+m)}_0 \left( \{ p_e \}_{e \in
L(\Lambda)} \right)\,.
\end{split}
\end{equation}
Here $\tau_e = \pm 1$, according to the orientation of the edge $e$.
We observe from (\ref{eq:contriLam}) that the momenta of the roots
of $\Lambda$ are the variables of the kernel of $J^{(k)}$, while the
momenta of the leaves of $\Lambda$ are the variables of the kernel
of $\gamma^{(k+m)}_0$ (this also explain why roots and leaves of
$\Lambda$ need to be paired).

\medskip

The denominators $(\a_e - p_e^2 + i \tau_e \mu_e)^{-1}$ are called
propagators; they correspond to the free evolutions in the expansion
(\ref{eq:xi}) and they enter the expression (\ref{eq:contriLam})
through the formula
\[ e^{itp^2} = \int^{\infty}_{-\infty} \rd \a \; \frac{e^{it(\a+i\mu)}}{\a- p^2
+i\mu} \] (here and in (\ref{eq:contriLam}) the measure $\rd \a$ is
defined by $\rd \a = \rd' \a / (2\pi i)$ where $\rd'\a$ is the
Lebesgue measure on $\bR$).

\medskip

The regularization factors $\mu_e$ in (\ref{eq:contriLam}) have to
be chosen such that $\mu_{\text{father}} = \sum_{e=\text{ son}}
\mu_{e}$ at every vertex. The delta-functions in
(\ref{eq:contriLam}) express momentum and frequency conservation
(the sum over $e \in v$ denotes the sum over all edges adjacent to
the vertex $v$; here $\pm \a_e = \a_e$ if the edge points towards
the vertex, while $\pm \a_e=-\a_e$ if the edge points out of the
vertex, and analogously for $\pm p_e$).

\medskip

An analogous expansion can be obtained for the error term
$\eta^{(k)}_{n,t}$ in (\ref{eq:eta}). The problem now is to analyze
the integral (\ref{eq:contriLam}) (and the corresponding integral
for the error term). Through an appropriate choice of the
regularization factors $\mu_e$ one can extract the time dependence
of $K_{\Lambda,t}$ and show that
\begin{equation}\label{eq:bound1}
\begin{split}
|K_{\Lambda,t}| \leq C^{k+m} \; t^{m/4}  \int & \prod_{e \in
E(\Gamma)} \frac{\rd \a_e \rd p_e}{\langle \a_e - p_e^2 \rangle} \;
\prod_{v \in V(\Gamma)} \delta \left( \sum_{e \in v} \pm \a_e
\right) \delta \left( \sum_{e \in v} \pm p_e \right) \, \\ &\times
\Big|\wh J^{(k)} \left( \{ p_e \}_{e \in R(\Gamma)} \right) \Big| \;
\Big| \wh \gamma^{(k+m)}_0 \left( \{ p_e \}_{e\in L
(\Gamma)}\right)\Big|
\end{split}
\end{equation}
where we introduced the notation $\langle x \rangle =
(1+x^2)^{1/2}$.

\medskip

Because of the singularity of the interaction at zero, we may be
faced here with an ultraviolet problem; we have to show that all
integrations in (\ref{eq:bound1}) are finite in the regime of large
momenta and large frequency. Because of (\ref{eq:aprit0}), we know
that the kernel $\wh \gamma^{(k+m)}_0 (\{ p_e \}_{e \in
L(\Lambda)})$ in (\ref{eq:bound1}) provides decay in the momenta of
the leaves. {F}rom (\ref{eq:aprit0}) we have, in momentum space,
\[ \int \rd p_1 \dots \rd p_{n} \; (p_1^2+1) \dots (p_{n}^2 + 1) \,
\wh \gamma^{(n)}_0 (p_1, \dots ,p_{n}; p_1, \dots ,p_{n}) \leq C^n
\] for all $n \geq 1$. Power counting implies that
\begin{equation}\label{eq:bound2} |\wh \gamma_{0}^{(k+m)} (\{ p_e \}_{e
\in L(\Lambda)})| \lesssim \prod_{e \in L(\Lambda)} \langle p_e
\rangle^{-5/2} \, .\end{equation} Using this decay in the momenta of
the leaves and the decay of the propagators $\langle \a_e - p_e^2
\rangle^{-1}, e \in E(\Lambda)$, we can prove the finiteness of all
the momentum and frequency integrals in (\ref{eq:contriLam}).
Heuristically, this can be seen using a simple power counting
argument. Fix $\kappa \gg 1$, and cutoff all momenta $|p_e| \geq
\kappa$ and all frequencies $|\a_e| \geq \kappa^2$. Each
$p_e$-integral scales then as $\kappa^3$, and each $\a_e$-integral
scales as $\kappa^2$. Since we have $2k+3m$ edges in $\Lambda$, we
have $2k+3m$ momentum- and frequency integrations. However, because
of the $m$ delta functions (due to momentum and frequency
conservation), we effectively only have to perform $2k+2m$ momentum-
and frequency-integrations. Therefore the whole integral in
(\ref{eq:contriLam}) carries a volume factor of the order
$\kappa^{5(2k+2m)} = \kappa^{10k + 10m}$. Now, since there are
$2k+2m$ leaves in the graph $\Lambda$, the estimate
(\ref{eq:bound2}) guarantees a decay of the order $\kappa^{-5/2 (2k
+ 2m)} = \kappa^{-5k-5m}$. The $2k+3m$ propagators, on the other
hand, provide a decay of the order $\kappa^{-2(2k+3m)} =
\kappa^{-4k-6m}$. Choosing the observable $J^{(k)}$ so that $\wh
J^{(k)}$ decays sufficiently fast at infinity, we can also gain an
additional decay $\kappa^{-6k}$. Since
\[ \kappa^{10k+10m} \cdot \kappa^{-5k-5m -4k -6m -6k} =
\kappa^{-m-5k} \ll 1 \] for $\kappa \gg 1$, we can expect
(\ref{eq:contriLam}) to converge in the large momentum and large
frequency regime. Remark the importance of the decay provided by the
free evolution (through the propagators); without making use of it,
we would not be able to prove the uniqueness of the infinite
hierarchy.

\medskip

This heuristic argument is clearly far from rigorous. To obtain a
rigorous proof, we use an integration scheme dictated by the
structure of the graph $\Lambda$; we start by integrating the
momenta and the frequency of the leaves (for which (\ref{eq:bound2})
provides sufficient decay). The point here is that when we perform
the integrations over the momenta of the leaves we have to propagate
the decay to the next edges on the left. We move iteratively from
the right to the left of the graph, until we reach the roots; at
every step we integrate the frequencies and momenta of the son edges
of a fixed vertex and as a result we obtain decay in the momentum of
the father edge. When we reach the roots, we use the decay of the
kernel $\wh J^{(k)}$ to complete the integration scheme. In a
typical step, we consider a vertex as the one drawn in Figure
\ref{fig:2} and we assume to have decay in the momenta of the three
son-edges, in the form $|p_e|^{-\lambda}$, $e=u,d,w$ (for some $2<
\lambda < 5/2$). Then we integrate over the frequencies $\a_u, \a_d,
\a_w$ and the momenta $p_u, p_d, p_w$ of the son-edges and as a
result we obtain a decaying factor $|p_r|^{-\lambda}$ in the
momentum of the father edge.
\begin{figure}
\begin{center}
\epsfig{file=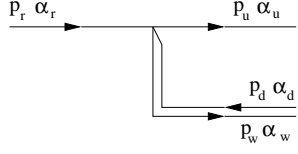,scale=.75}
\end{center}
\caption{Integration scheme: a typical vertex}\label{fig:2}
\end{figure}
In other words, we prove bounds of the form
\begin{equation}\label{eq:scheme}
\int  \frac{\rd \a_u \rd \a_d \rd \a_w \rd p_u \rd p_d \rd
p_w}{|p_u|^{\lambda} |p_d|^{\lambda} |p_w|^{\lambda}} \,
\frac{\delta (\a_r = \a_u +\a_d - \a_w)\delta (p_r = p_u +p_d
-p_w)}{\langle \a_u - p_u^2 \rangle \langle \a_d - p_d^2 \rangle
\langle \a_w -p_w^2 \rangle} \leq \frac{\const}{|p_r|^{\lambda}}\,.
\end{equation}
Power counting implies that (\ref{eq:scheme}) can only be correct if
$\lambda >2$. On the other hand, to start the integration scheme we
need $\lambda < 5/2$ (from (\ref{eq:bound2}) this is the decay in
the momenta of the leaves, obtained from the a-priori estimates). It
turns out that, choosing $\lambda = 2 + \eps$ for a sufficiently
small~$\eps>0$, (\ref{eq:scheme}) can be made precise, and the
integration scheme can be completed.

\medskip

After integrating all the frequency and momentum variables, from
(\ref{eq:bound1}) we obtain that
\[ |K_{\Lambda,t}| \leq C^{k+m} \; t^{m/4} \] for every $\Lambda \in
\cF_{m,k}$. Since the number of diagrams in $\cF_{m,k}$ is bounded
by $C^{k+m}$, it follows immediately that
\[ \left| \tr \; J^{(k)} \, \xi_{m,t}^{(k)} \right| \leq C^{k+m}
t^{m/4} \,. \] Note that, from (\ref{eq:xi}), one may expect
$\xi_{m,t}^{(k)}$ to be proportional to $t^m$. The reason why we
only get a bound proportional to $t^{m/4}$ is that we effectively
use part of the time integration to control the singularity of the
potentials.

\medskip

Note that the only property of $\gamma^{(k+m)}_0$ used in the
analysis of (\ref{eq:contriLam}) is the estimate (\ref{eq:aprit0}),
which provides the necessary decay in the momenta of the leaves.
However, since the a-priori bound (\ref{eq:apri}) hold uniformly in
time, we can use a similar argument to bound the contribution
arising from the error term $\eta^{(k)}_{n,t}$ in (\ref{eq:eta}) (as
explained above, also $\eta^{(k)}_{n,t}$ can be expanded analogously
to (\ref{eq:expa}), with contributions associated to Feynman graphs
similar to (\ref{eq:contriLam}); the difference, of course, is that
these contributions will depend on $\gamma^{(k+n)}_s$ for all $s \in
[0,t]$, while (\ref{eq:contriLam}) only depends on the initial
data). Thus, we also obtain \begin{equation}\label{eq:errorbound}
\left| \tr \; J^{(k)} \, \eta_{n,t}^{(k)} \right| \leq C^{k+n}\;
t^{n/4}\,. \end{equation} This bound immediately implies the
uniqueness. In fact, given two solutions $\Gamma_{1,t} = \{
\gamma_{1,t}^{(k)} \}_{k \geq 1}$ and $\Gamma_{2,t} = \{
\gamma_{2,t}^{(k)} \}_{ k \geq 1}$ of the infinite hierarchy
(\ref{eq:GPhier}), both satisfying the a-priori bounds
(\ref{eq:apri}) and with the same initial data, we can expand both
in a Duhamel series of order $n$ as in (\ref{eq:Duhamel}). If we fix
$k \geq 1$, and consider the difference between $\gamma^{(k)}_{1,t}$
and $\gamma^{(k)}_{2,t}$, all terms (\ref{eq:xi}) cancel out because
they only depend on the initial data. Therefore, from
(\ref{eq:errorbound}), we immediately obtain that, for arbitrary
(sufficiently smooth) compact $k$-particle operators $J^{(k)}$,
\[ \left|\tr J^{(k)} \left( \gamma_{1,t}^{(k)} - \gamma_{2,t}^{(k)}
\right)\right| \leq 2 \, C^{k+n} \; t^{n/4} \] Since it is
independent of $n$, the left side has to vanish for all $t < 1/C^4$.
This proves uniqueness for short times. But then, since the a-priori
bounds hold uniformly in time, the argument can be repeated to prove
uniqueness for all times.

\end{document}